\documentclass[runningheads]{llncs}
\usepackage{graphicx}
\usepackage{textcomp, gensymb}
\usepackage{gensymb}
\usepackage{amsmath, amssymb}
\usepackage{support-caption}
\usepackage{subcaption}
\newtheorem{obs}{Observation}
\usepackage{enumitem}

\newlist{mylistenv}{enumerate}{3}
\newenvironment{mylist}[1]{
    \setlist[mylistenv]{ leftmargin  = 2\parindent, label=#1\arabic{mylistenvi},ref=#1\arabic{mylistenvi}} 
    \setlist[mylistenv,2]{label=#1\arabic{mylistenvi}.\arabic{mylistenvii},ref=#1\arabic{mylistenvi}.\arabic{mylistenvii}}
    \setlist[mylistenv,3]{label=#1\arabic{mylistenvi}.\arabic{mylistenvii}.\arabic{mylistenviii}.,ref=#1\arabic{mylistenvi}.\arabic{mylistenvii}.\arabic{mylistenviii}}
    \renewenvironment{mylist}{\begin{mylistenv}}{\end{mylistenv}}
    \begin{mylistenv}
    }{
    \end{mylistenv}
}
\usepackage{soul}
\usepackage[noend, ruled, lined, linesnumbered, commentsnumbered,vlined, longend]{algorithm2e}
\usepackage{algcompatible}
\usepackage{xcolor}

\newcommand{\Rec}{\textup{Rec}}
\begin{document}

\title{ On Geometric Shape Construction via Growth Operations}

\author{Nada Almalki
\and
Othon Michail
\orcidID{0000-0002-6234-3960} 
}

\authorrunning{N. Almalki and O. Michail}

\institute{Department of Computer Science, University of Liverpool, UK\\
\email{\{N.Almalki, Othon.Michail\}@liverpool.ac.uk}\\
}

\maketitle            

\begin{abstract}
In this work, we investigate novel algorithmic growth processes. In particular, we propose three growth operations, \emph{full doubling}, \emph{RC doubling} and \emph{doubling}, and explore the algorithmic and structural properties of their resulting processes under a geometric setting. In terms of modeling, our system runs on a 2-dimensional grid and operates in discrete time-steps. The process begins with an initial shape $S_I=S_0$ and, in every time-step $t \geq 1$, by applying (in parallel) one or more growth operations of a specific type to the current shape-instance $S_{t-1}$, generates the next instance $S_t$, always satisfying $|S_t| > |S_{t-1}|$.
Our goal is to characterize the classes of shapes that can be constructed in $O(\log n)$ or polylog $n$ time-steps and determine whether a final shape $S_F$ can be constructed from an initial shape $S_I$ using a finite sequence of growth operations of a given type, called a \emph{constructor of $S_F$}.

For \emph{full doubling}, in which, in every time-step, every node generates a new node in a given direction, we completely characterize the structure of the class of shapes that can be constructed from a given initial shape. For \emph{RC doubling}, in which complete columns or rows double, our main contribution is a linear-time centralized algorithm that for any pair of shapes $S_I$, $S_F$ decides if $S_F$ can be constructed from $S_I$ and, if the answer is yes, returns an $O(\log n)$-time-step constructor of $S_F$ from $S_I$. For the most general \emph{doubling} operation, where up to individual nodes can double, we show that some shapes cannot be constructed in sub-linear time-steps and give two universal constructors of any $S_F$ from a singleton $S_I$, which are efficient (i.e., up to polylogarithmic time-steps) for large classes of shapes. Both constructors can be computed by polynomial-time centralized algorithms for any shape $S_F$.

\keywords{Centralized algorithm \and Geometric growth operations \and Programmable matter \and Constructor.}
\end{abstract}

\section{Introduction }
\label{sec:intro}
The realization that many natural processes are essentially algorithmic, has fueled a growing recent interest in formalizing their algorithmic principles 
and in developing new algorithmic approaches and technologies inspired by them. Examples of algorithmic frameworks inspired by biological and chemical systems are population protocols~\cite{angluin2006computation,angluin2007computational,michail2016simple}, ant colony optimization~\cite{cornejo2014task,doty2012theory}, DNA self-assembly~\cite{doty2012theory,rothemund2006folding,rothemund2000program,woods2019diverse}, and the algorithmic theory of programmable matter~\cite{akitaya2021universal,almethen2020pushing,michail2019transformation}.

Motivated by these advancements and by principles of biological development which are apparently algorithmic, we introduce a set of geometric growth processes and study their algorithmic and structural properties. These processes start from an initial shape of nodes $S_I$, possibly a singleton, and by applying a sequence of growth operations eventually develop into well-defined global geometric structures. The considered growth operations involve at most one new node being generated by any existing node in a given direction and the resulting reconfiguration of the shape as a consequence of a set of nodes being generated within it. This node-generation primitive is also inspired by the self-replicating capabilities of biological systems, such as cellular division, their higher-level processes such as embryogenesis \cite{chan19molecular}, and by the potential of the future development of self-replicating robotic systems.

In a recent study, Mertzios \textit{et al.}~\cite{mertzios2021complexity} investigated a network-growth process at an abstract graph-theoretic level, free from geometric constraints. Our goal here is to study similar growth processes under a geometric setting and show how these can be fine-tuned to construct interesting geometric shapes efficiently, i.e., in polylogarithmic time-steps. Aiming to focus exclusively on the effect of growth operations, apart from local growth we do not allow any other form of shape reconfigurations. Preliminary such growth processes, mostly for rectangular shapes, were developed by Woods \textit{et al.}~\cite{woods2013active}. Their approach was to first grow such shapes in polylogarithmic and to then transform them into arbitrary geometric shapes and patterns through additional reconfiguration operations, the latter essentially capturing properties of molecular self-assembly systems. Like them, we study the problem of constructing a desired final shape $S_F$ starting from an initial shape $S_I$ via a sequence of shape-modification operations. However, in this work the considered operations are only local growth operations. To the best of our knowledge, the structural characterization and the underlying algorithmic complexity of constructing geometric shapes by growth operations, have not been previously considered as problems of independent interest.

\subsection{Our Approach and Contribution}\label{subsec:our-approach}

In this work, our main objective is to study \emph{growth operations} in a centralized geometric setting. Applying a sequence of such operations in a centralized way, yields a centralized geometric growth process.
Our model can be viewed as an applied, geometric version of the abstract network-growth model of Mertzios \textit{et al.}~\cite{mertzios2021complexity}. 
The considered model is discrete and operates on a 2D square grid. 

Connectivity preservation is an essential aspect of both biological and of the so-inspired robotic and programmable systems, because it allows the system to maintain its strength and coherence and enables sharing of resources between devices in the system. In light of this, all the shapes discussed in this work are assumed to be connected and the considered growth operations cannot break the shape's connectivity.
For all types of considered operations, the study revolves around the following main questions: (i) \emph{What is the class of shapes that can be constructed efficiently from a given initial shape via a sequence of growth operations?} (ii) \emph{Is there a polynomial-time centralized algorithm that can decide if a given target shape $S_F$ can be constructed from a given initial shape $S_I$ and, whenever the answer is positive, return an efficient constructor of $S_F$ from $S_I$?} 

The growth operations considered in this paper are characterized by the following additional properties:
\begin{itemize}
    \item In general, more than one growth operation can be applied at the same time-step (parallel version). To simplify the exposition of some of our results and without losing generality, we shall sometimes restrict attention to a single operation per time step (sequential version).
    \item To avoid having to deal with colliding operations, we restrict attention to single-direction growth operations. That is, for each time-step $t$, a direction $d \in \{$north, east, south, west$\}$ is fixed and any operation at $t$ must be in direction $d$.
    For clarity of presentation of the results in this work, we shall focus mostly on the \emph{east} and \emph{north} directions of the considered operations. Due to the nature of these operations, generalizing to all four directions is immediate.
\end{itemize}

We study three growth operations, \emph{full doubling}, \emph{RC doubling}, and \emph{doubling}, where full doubling is the most restricted and doubling the most general one. In \emph{full doubling}, in every time-step, every node generates a new node in a given direction, in \emph{RC doubling}, complete columns or rows double, and in \emph{doubling} up to individual nodes can double.

For \emph{full doubling}, we completely characterize the structure of the class of shapes that are reachable from any given initial shape. For \emph{RC doubling}, our main contribution is a linear-time centralized algorithm that for any pair of shapes $S_I$, $S_F$ decides if $S_F$ can be constructed from $S_I$ and, if the answer is yes, returns an $O(\log n)$-time-step constructor of $S_F$ from $S_I$. For \emph{doubling}, we show that some shapes cannot be constructed in sub-linear time-steps and give two universal constructors of any $S_F$ from a singleton $S_I$, which are efficient (i.e., up to polylogarithmic time-steps). 
for large classes of shapes. Both constructors can be computed by polynomial-time centralized algorithms for any shape $S_F$.\footnote{Note that there are two distinct notions of time used in this paper. One represents the time-steps of a growth process, while the other represents the running time of a centralized algorithm deciding reachability between shapes and returning constructors for them. We shall always distinguish between the two by calling the former \emph{time-steps} and the latter \emph{time}.}

In Section~\ref{subsec:related-work}, we discuss the related literature. Section~\ref{sec:Models and Preliminaries} presents all definitions that are used throughout the paper.
Sections~\ref{sec:full-doubling}, \ref{sec:RC-doubling}, and \ref{sec:doubling} present our results for \emph{full doubling}, \emph{RC doubling}, and \emph{doubling}, respectively. Finally, in Section~\ref{sec:conclusion}, we also give further research directions opened by our work.

\subsection{Related work}
\label{subsec:related-work}
Recent work has focused on studying the algorithmic principles of reconfiguration, with the potential of developing artificial systems that will be able to modify their  physical  properties, such  as  reconfigurable  robotic ensembles and self-assembly systems. For example, the area of algorithmic self-assembly of DNA aims to understand how to train molecules to modify themselves while also controlling their own growth~\cite{doty2012theory}. Several theoretical models of programmable matter have been developed, including DNA self-assembly and other passively dynamic models~\cite{doty2012theory,michail2018terminating} as well as models enriched with active molecular components~\cite{woods2013active}.

One example of a geometric programmable matter model, which is presented in~\cite{derakhshandeh2014brief}, is known as the \textit{Amoebot} and is inspired by amoeba behavior. In particular, programmable matter is modeled as a swarm of distributed autonomous self-organizing entities that operate on a triangular grid. Research on the Amoebot model has made progress on understanding its computational power and on developing algorithms for basic reconfiguration tasks such as coating~\cite{derakhshandeh2017universal} and shape formation~\cite{derakhshandeh2016universal,di2020shape}. Other authors have investigated cycle-shaped programmable matter modules that can rotate or slide a device over neighboring devices through an empty space~\cite{dumitrescu2004pushing,dumitrescu2004formations,michail2019transformation,connor2021centralised}, with the goal of capturing the global reconfiguration capabilities of local mechanisms that are feasible to be implemented with existing technology. The authors in~\cite{michail2019transformation} proved that the decision problem of transformation between two shapes is in \textbf{P}. In addition, another recent research work~\cite{almethen2020pushing} investigated a linear-strength mechanism through which a node can push a line of one or more nodes by one position in a single time-step. Other linear-strength mechanisms are the one by Woods \textit{et al.}~\cite{woods2013active}, where a node can rotate a whole line of connected nodes, simulating arm rotation, or the one by Aloupis \emph{et al.} \cite{aloupis2008reconfiguration} on crystalline robots, equipped with powerful lifting capabilities where a single robot can lift by one position a line of other robots.

A recent study in the field of highly dynamic networks, which is presented in~\cite{mertzios2021complexity}, is partially inspired by the abstract-network approach followed in~\cite{michail2020distributed}. The authors completely disregard geometry and develop a network-level abstraction of programmable matter systems. Their model starts with a single node and grows the target network $G$ to its full size by applying local operations of node replication. Local edges are only activated upon a node's generation and can be deleted at any time but contribute negatively to the edge-complexity of the construction. The authors develop centralized algorithms that generate basic graphs such as paths, stars, trees, and planar graphs and prove strong hardness results. We similarly focus on centralized structural and algorithmic characterizations as a first step that will promote our understanding of such novel models and will facilitate the future development of more applied constructions, like fully distributed ones.

\section{Model and Preliminaries}\label{sec:Models and Preliminaries}

The programmable matter systems considered in this paper operate on a 2-dimensional square grid. Each grid position (cross point) is identified  by its $x$ and $y$ coordinates, $x \geq 0$ representing the row and $y \geq 0$ the column. Systems of this type consist of $n$ nodes that form a connected shape $S$, as in Fig.~\ref{fig:Shape_S}. Each node $u$ of shape $S$ is represented by a circle occupying a position on the grid. Time consists of discrete time-steps and in every time-step $t\geq 0$, zero or more growth operations can occur depending on the type of operation considered. At any given time-step $t$, each node $u \in S$ is determined by its coordinates $(u_x, u_y)$ and no two nodes can occupy the same position at the same time-step. Two distinct nodes $u=(u_x, u_y)$ and $v=(v_x, v_y)$ are \emph{neighbors} if $u_x\in\{v_x-1,v_x+1\}$ and $u_y=v_y$ or $u_y\in\{v_y-1,v_y+1\}$ and $u_x=v_x$, that is, if they are at orthogonal distance one from each other. In that case, we are assuming that, unless explicitly removed, a connection (or edge) $uv$ exists between $u$ and $v$.

\begin{figure}[ht]
\centering 
\includegraphics[width=1.4in]{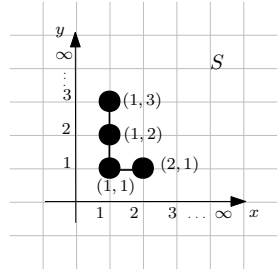}
\caption{An example of a shape $S$ with four nodes $\{(1,1), (1,2),(1,3),(2,1)\}$.}
\label{fig:Shape_S}
\end{figure}

\begin{definition}[Row and Column of a Shape]\label{def:row-column-shape}
A row (respectively column) of a shape $S$ is the set of all nodes of $S$ with the same fixed $y$-coordinate (respectively $x$-coordinate).

By $S_{\cdot, i}$ we denote the row of $S$ consisting of all nodes whose $y$-coordinate is $i$, i.e., $S_{\cdot,i} = \{(x,i) ~|~ (x,i) \in S\}$. 
Similarly, \emph{column} $j$ of $S$, denoted as $S_{j,\cdot}$, is the set of all nodes of $S$ whose $x$-coordinate is $j$, i.e., $S_{j,\cdot} = \{(j,y) ~|~ (j,y) \in S\}$. When the shape $S$ is clear from context, we will refer to $S_{\cdot, i}$ as $R_i$ and to $S_{j,\cdot}$ as $C_j$.
\end{definition}

\begin{definition}[Translation Operation]
Given a set of integer points $Q$, the north (south) $k$-\emph{translation} of $Q$ is defined as $\uparrow_{k} Q =\{(x,y+k) ~|~ (x,y)\in Q\}$ ($\downarrow_{k} Q$, similarly defined). The east (west) $l$-\emph{translation} of $Q$ is defined as $\stackrel{l}\rightarrow Q = \{(x+l,y) ~|~ (x,y)\in Q\}$ ($\stackrel{l}\leftarrow Q$, similarly defined). 
\end{definition}

\begin{definition}[Rigid Connection]\label{def:rigid-connection}
A connection $uv$ between two nodes $u$ and $v$ of a shape $S$ is
\emph{rigid} if and only if a $1$-translation of one node in any direction $d$ implies a $1$-translation of the other in the same direction, unless $uv$ is first removed.
\end{definition}
Throughout, all connections are assumed to be rigid.

The basic concept of growth operation is that a node $u \in S_t $ generates a new node $u' \in S_{t+1}$. In particular, we are exploring three specific growth operations a \emph{full doubling}, \emph{row and column doubling} and \emph{doubling}. In most cases, every node $u \in S_t$ is colored black while its generations are colored gray at the next time-step $t+1$ after any type of growth operations.
Furthermore, any type of growth operation $o$ is equipped with a \emph{linear-strength} mechanism, which is the ability of a generated node $u'$ to translate its connected component on a growing direction in a single time-step.

Throughout this paper, $l, k$ will represent the total number of horizontal and vertical growth operations performed (respectively). By horizontal, the direction $d$ is either \emph{east} or \emph{west}, while the vertical $d$ is either \emph{north} or \emph{south}.

\begin{definition}[Growth Operation]\label{def:operation}
A \emph{growth operation} $o$ is an operation that when applied on a shape instance $S_t$, for all time-steps $t\geq 0$, yields a new shape instance $S_{t+1}=o(S_t)$, such that $|S_{t+1}|>|S_t|$.
\end{definition} 

In this work, we consider three specific types of growth operations moving from the most special to the most general: \emph{full doubling }, \emph{RC doubling} and \emph{doubling} operation.
For the sake of clarity, we will provide a high-level overview of these three operations, with the more technical versions appearing in their respective Sections~\ref{sec:full-doubling},~\ref{sec:RC-doubling} and~\ref{sec:doubling}. First, a \emph{full doubling} operation is a growth operation in which every node $(u \in S_t)$ generates a new node $(u'\in S_{t+1})$, that is, $|S_{t+1}| = 2|S_t|$. Then, \emph{row and column doubling} denoted by \emph{RC doubling}, is a growth operation where in each time-step $t$, a subset of columns (rows) is selected and these are fully doubled. Finally, the most general version of these operations is a \emph{doubling} operation, in which, in each time-step $t$, any subset of the nodes can double in a given direction. 

The differences between these three operations are highlighted in Fig.~\ref{fig:diff-oper}.

\begin{figure}[ht]
\centering 
\includegraphics[width=4in]{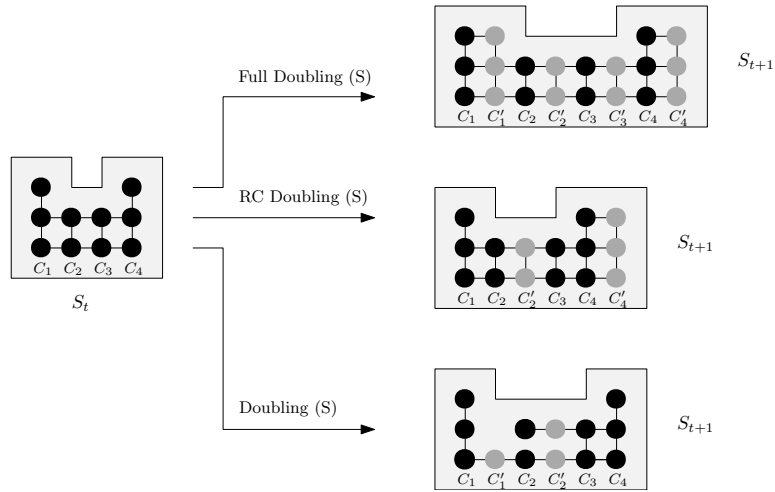}
\caption{Illustration of the results $S_{t+1}$ when applying different growth operations on the same shape $S_{t}$ where the direction of growing is east.}
\label{fig:diff-oper}
\end{figure}

\begin{definition}[Reachability Relation]\label{def:reachability}
Given a growth operation of a given type, we define a reachability relation $\rightsquigarrow$ on pairs of shapes $S$, $S'$ as follows. $S\rightsquigarrow S'$ iff there is a finite sequence $\sigma=o_1,o_2,\ldots,o_{t_{last}}$ of operations of a given type for which $S=S_0,o_1,S_1,o_2,S_2\ldots$ $,S_{({t_{last}}-1)},o_{t_{last}},S_{t_{last}}=S'$, where $S_i=o(S_{i-1})$ for all $1\leq i\leq t_{last}$. Whenever we want to emphasize a particular such sequence $\sigma$, we write $S \stackrel{\sigma}\rightsquigarrow S'$ and say that $\sigma$ \emph{constructs} shape $S'$ from $S$.
\end{definition}

\begin{definition}[Constructor]\label{def:cons-sch}
A \emph{constructor} $\sigma = (o_1, o_2,\ldots,$ $o_{t_{last}})$ is a finite sequence of doubling operations of a given type, $1\leq i\leq t_{last}$. 
\end{definition}

\begin{remark} \label{rem:sd-doubling-process}
Note that in Definition \ref{def:cons-sch} the directions of different $o_i$'s do not need to be the same.
\end{remark}

\subsection{Problem Definitions}\label{subsec:prob-def}
We now formally define the problems to be considered in this paper.\\

\noindent\emph{Class characterization}. Identify the family of shapes $S_F$ that can be obtained from a given initial connected shape $S_I$ via a sequence of growth operations of a given type.\\

\noindent{\sc ShapeConstruction}. Given a pair of shapes ($S_I ,S_F$) decide if $S_I\rightsquigarrow S_F$. If yes, compute a sequence $\sigma$ that constructs $S_F$ from $S_I$. In the special case of this problem in which $S_I$ is by assumption a singleton, we shall assume that the input is just $S_F$.

\section{Full Doubling }\label{sec:full-doubling}
In this section, after providing a formal definition of the full doubling operation, we investigate the class characterization problem under this operation. Note that when the initial shape consists of a single node, i.e., $|S_I| = 1$, the characterization is straightforward which is in Section~\ref{subsection:case_initial_equal1}. Section~\ref{subsection:case_initial_more_than_one_node} discusses the more general case where $|S_I| \geq 1$.

\begin{definition}[Full Doubling]\label{def:full-doubling}
After applying a \emph{full doubling operation} on $S$ a \emph{new shape} $S^{\prime}$ is obtained, depending only on the direction $d$ of the operation:

\begin{enumerate}
    \item If the direction $d$ of a \emph{full doubling operation} is an \emph{east}, then for every column $S_{j,\cdot}$ of $S$ a new column is generated to the \emph{east} of $S_{j,\cdot}$. The effect of applying this to all columns is that every column $S_{j,\cdot}$ of $S$ is translated to the east by $j-1$, such that $S^\prime_{2j-1,\cdot} = \stackrel{j-1}\rightarrow S_{j,\cdot}$, and generates the new column $S^\prime_{2j,\cdot} = \stackrel{j}\rightarrow S_{j,\cdot}$. Therefore, the new shape $S^{\prime}$ of this doubling operation is $S^{\prime}=\bigcup_{j} (S^\prime_{2j-1,\cdot}\cup~ S^\prime_{2j,\cdot})$. 
  
    \item If the direction $d$ of a \emph{full doubling operation} is a \emph{north}, then for every row $S_{\cdot,i}$ of $S$ a new row is generated to the north of $S_{\cdot,i}$. The effect of applying this to all rows is that every row $S_{\cdot,i}$ of $S$ is translated to the north by $i-1$, such that $ S^\prime_{\cdot,2i-1} = \uparrow_{i-1} S_{\cdot,i}$, and generates the new row $S^\prime_{\cdot,2i} = \uparrow_{i} S_{\cdot,i}$. Therefore, the new shape $S^{\prime}$ of this doubling operation is $S^{\prime}=\bigcup_{i} (S^\prime_{\cdot,2i-1}\cup~ S^\prime_{\cdot,2i})$.

\end{enumerate}
\end{definition}
In other words, if a full doubling operation is performed on $S$ in the east direction, then a set of columns equal to the original is generated. Every original column is translated by the number of original columns to its west and its own copy is generated to the east of its final position. Similarly for rows.

Fig.~\ref{fig:column-row-coordinates} depicts an example where $S$ initially consists of two columns  $S_{1,\cdot}$, $S_{2,\cdot}$, after applying a full doubling operation in the east direction
this results in $S'$ where $S_{1,\cdot}$ generates $S_{2,\cdot}$ and $S_{2,\cdot}$ is translated by $1$ to the east and becomes $S_{3,\cdot}$ which creates $S_{4,\cdot}$.

\begin{figure}[ht]
\centering 
\parbox{3.5cm}{
\includegraphics[width=3.5cm]{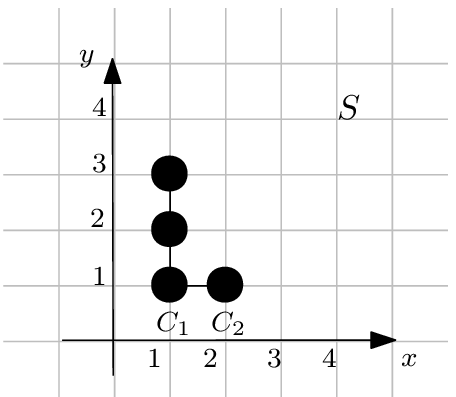}
\subcaption{Initial shape $S$.}}
\qquad
\qquad
\parbox{3.5cm}{
\includegraphics[width=3.5cm]{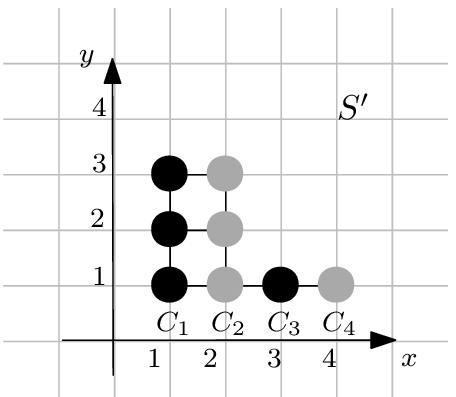}
\subcaption{$S'$ after a full doubling operation.}}
\caption{Illustration example of the new shape $S'$ after applying a full doubling operation $o$ on the columns of $S$, in the east direction.}
\label{fig:column-row-coordinates}
\end{figure}

\subsection{Single-node Initial Shape $|S_I| =1$}\label{subsection:case_initial_equal1}
This section characterizes shapes that can be obtained by a \emph{full doubling} operation that starts with an initial shape of a single node $|S_I| = 1$. The characterization is straightforward and its main purpose is to illustrate the dynamics of the \emph{full doubling }operation.

\begin{definition}[Cartesian Product]\label{def:cartesian-product}
For two sets $X$ and $Y$ we denote by $X \times Y$ the Cartesian product of $X$ and $Y$.
\end{definition}

\begin{definition}[Rectangle]\label{def:rectangle}
A \emph{$(p,q)$-rectangle} is a set of the form $A \times B$, where $A$ and $B$ are sets of \emph{consecutive} integers of cardinally $p$ and $q$ respectively.
A set is called \emph{rectangle} if it is a $(p,q)$-rectangle for some natural numbers $p$ and $q$. Given an integer point $u$ and two natural numbers $p$ and $q$, the \emph{$(p,q)$-rectangle around $u$}, denoted $\Rec(u,p,q)$, is the $(p,q)$-rectangle $\{u_x, u_x+1, \ldots, u_x+p-1\} \times \{u_y, u_y+1, \ldots, u_y+q-1\}$.
\end{definition}

\begin{proposition} \label{prop:sd-doubling-process}
Let $l$ and $k$ be natural numbers. A full doubling constructor with $l$ horizontal and $k$ vertical full doubling operations that starts from a single node constructs a rectangle of size $ 2^{l} \times 2^k$.
\end{proposition}

\begin{proof}
We will prove the statement by induction on the number $t = l + k$ of operations.
For $t = 0$ the statement is trivial.

Let $t > 0$ and assume that the proposition holds for all constructors with less than $t$ full doubling operations.

Let us fix a full doubling constructor $\sigma = (o_1, o_2, \ldots, o_t)$, with exactly $t$ operations and let 
$l$ and $k$ be the number of horizontal and vertical doubling in this constructor respectively. Assume that the last operation in $\sigma$ is 
a horizontal full doubling and consider the constructor $\sigma' = (o_1, o_2, \ldots, o_{t-1})$ consisting of the first $t-1$ operations of $\sigma$.
By the induction hypothesis, constructor $\sigma'$ constructs a rectangle of
size $2^{l} \times 2^{k}$. The final shape $S_F$ constructed by $\sigma$ is obtained from the shape $S'$ constructed by $\sigma'$ by applying a horizontal doubling operation. Hence $\sigma$ constructs a rectangle of size $2^{l} \times 2^{k}$, as desired.

The case when the last operation of $\sigma$ is a vertical doubling is proved similarly, let assume that the last full doubling operation of $\sigma$ is a vertical doubling and consider that the constructor $\sigma' = (o_1, o_2, \ldots, o_{t-1})$ consisting of the first $t-1$ operations of $\sigma$.

By the induction hypothesis, constructor $\sigma'$ constructs a rectangle of
size $2^{l} \times 2^{k-1}$. The final shape $S_F$ constructed by constructor $\sigma$ is obtained from the shape $S'$ constructed by $\sigma'$ by applying a vertical doubling operation in this case, thus,  $\sigma$ constructs a rectangle of size $2^{l} \times 2^{k}$, as desired (see Fig.~\ref{fig:doubling-example}).
\qed 
\end{proof}

\begin{figure}[ht]
\centering 
\includegraphics[width=3in]{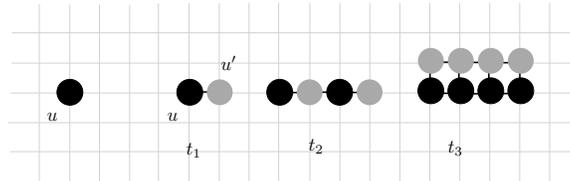}
\caption{Construction overview of a rectangle where $n = (2^2 \times 2 )$.}
\label{fig:doubling-example}
\end{figure}

\begin{remark} \label{remark1}
From Proposition~\ref{prop:sd-doubling-process}, it follows that the final shape of a full doubling constructor depends on the number of horizontal and vertical operations and not on the order of these operations.
\end{remark}

\subsection{An Arbitrary Connected Initial Shape $|S_I|\geq 1$}
\label{subsection:case_initial_more_than_one_node}
In this section we characterize shapes that can be obtained by a sequence of full doubling operations from an arbitrary connected initial shape $S_I$, where $|S_I| \geq 1$.

\begin{definition}[$w(C_j,u)$]\label{def:w(C_j,u)}
Let $w(C_j,u)$ denote the number of columns to the left (west) of a node $u$ in a column $j$, that is, $w(C_j,u)= u_x - C_l$, where $C_l$ is the leftmost column.
\end{definition}

\begin{definition}[$s(R_i,u)$]\label{def:s(R_i,u)}
Let $s(R_i,u)$ denote the number of rows below (south) of a node $u$ in a row $i$, that is, $s(R_i,u)= u_y - R_b$, where $R_b$ is the bottom-most row.
\end{definition}

\begin{definition}[Reconfiguration Function]\label{def:reconf-func}
Given two integers $l,k >0$, we define a reconfiguration function $F_{l,k}$ that maps a shape to another shape as follows:
 
 \begin{enumerate}
     \item First, the coordinates of $|S|$ points of $F_{l,k}(S)$ are determined as a function of the coordinates of the points of $S$. For each $u\in S$ the coordinates of $u^\prime\in F_{l,k}(S)$ are given by $( u_x + (2^l-1)w(C_j,u), u_y + (2^k-1)s(R_i,u) )$, as in Fig.~\ref{fig:reconf-func-scenario} (a and b).
     \item Generate the Cartesian product (Definition~\ref{def:cartesian-product})around $u^\prime$
     such that, $Rec(u',$ $2^l, 2^k) = \{u'_x+1,\ldots, u'_x+(2^l-1)\} \times \{u'_y+1,\ldots, u'_y+(2^k-1)\}$ originating at $u^\prime$. 
     Adding all points of these rectangles to $F_{l,k}(S)$ completes the definition of $F_{l,k}(S)$.
 \end{enumerate}

The output of the \emph{reconfiguration function} after these two phases is a shape $S$, such that $F_{l,k}(S)=\underset{u\in S} \bigcup Rec(u',2^l,2^k) $, as presented in Fig.~\ref{fig:final-shape} (note that $u'$ is a function of $u$ in the union).

\end{definition}

\begin{remark} \label{rem:reconf-func}
Another comprehensive definition of \emph{reconfiguration function} is that given two integers $l,k >0$, a \emph{reconfiguration function $F_{l,k}$} maps an input shape $S$ to the output shape $F_{l,k}(S)$ as follows:\\
For every node $u \in S$, let $u'$ be the point $( u_x + (2^l-1)w(C_j,u), u_y + (2^k-1)s(R_i,u))$. Then,
$$ F_{l,k}(S):=\underset{u\in S} \bigcup \Rec(u',2^l,2^k).
$$
\end{remark}

\begin{figure}[ht]
\centering
\parbox{4cm}{
\includegraphics[width=4cm]{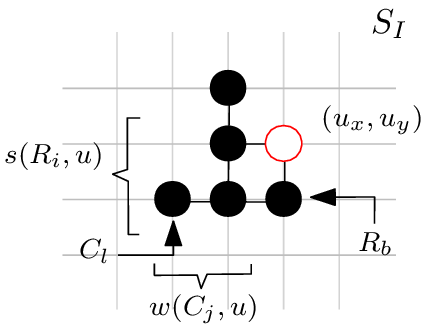}
\subcaption{ Input shape $S$.}}
\qquad
\parbox{6.8cm}{
\includegraphics[width=6.8cm]{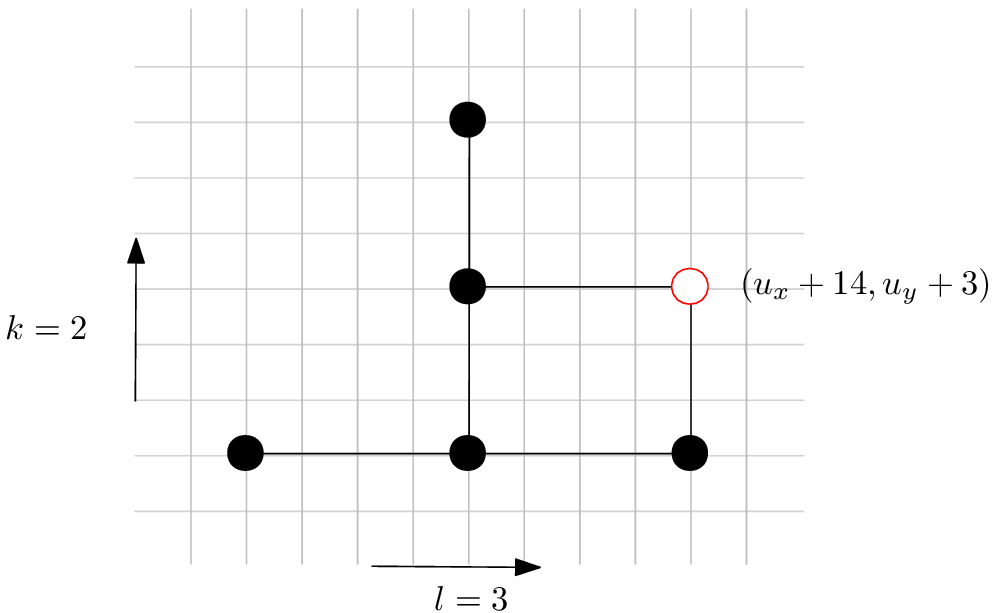}
\subcaption{After step 1 of \emph{reconfiguration function}.}}
\caption{Step 1 of a \emph{reconfiguration function} for $l=3$ and $k=2$. Node $u=(u_x,u_y)\in S$ (red node in (a)) will give a new node $u^\prime=(u^\prime_{x}+14,u^\prime_{y}+3)\in F_{3,2}(S)$ (red node in (b)) and similarly for the other nodes.}
\label{fig:reconf-func-scenario}
\end{figure}

\begin{figure}[ht]
\centering 
\includegraphics[width=4.5in]{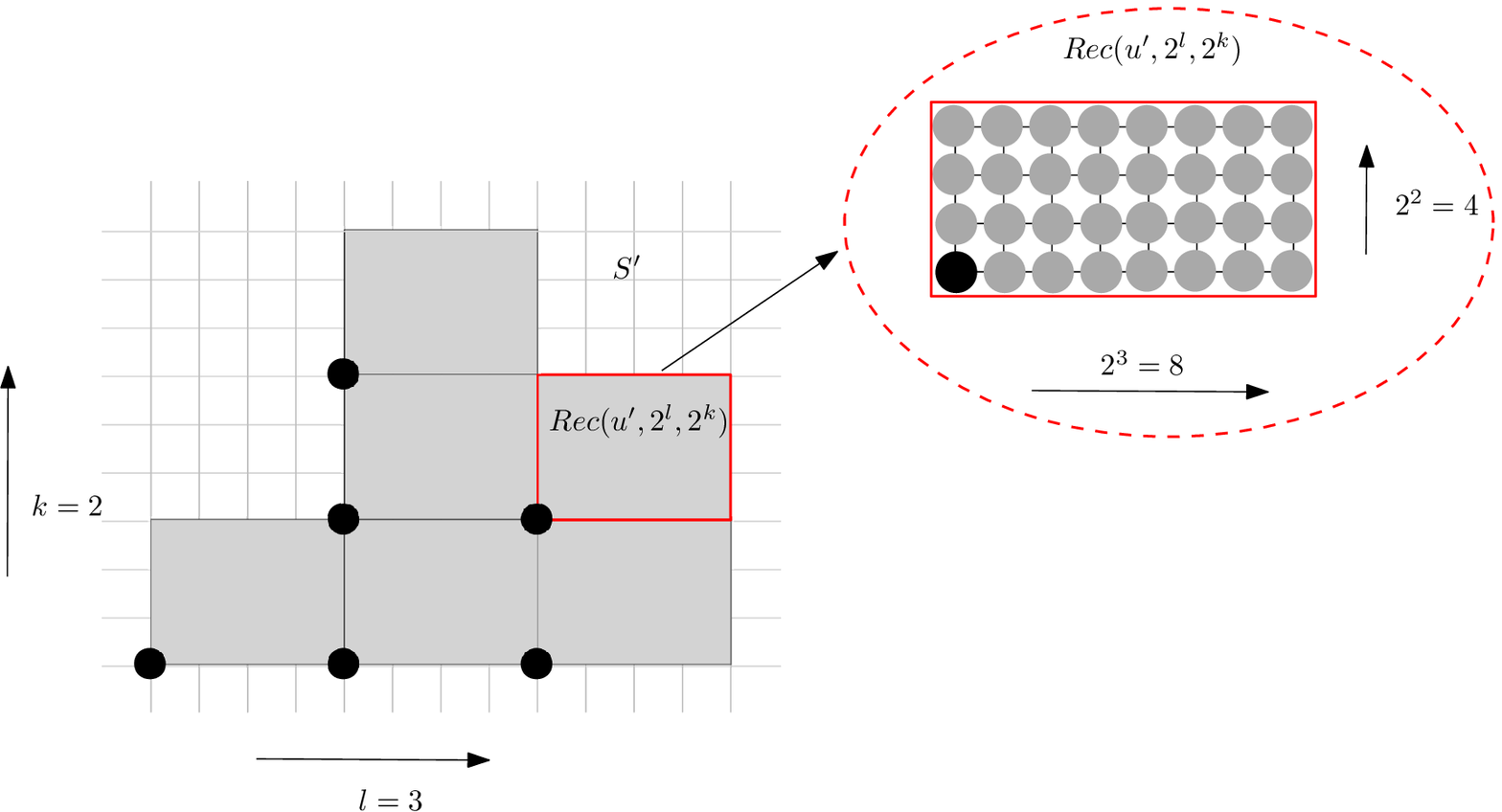}
\caption{An example of the output shape $S'$ after applying the \emph{reconfiguration function} $F_{l,k}(S)$.}
\label{fig:final-shape}
\end{figure}

\begin{lemma}[Additivity of Reconfiguration Function]\label{lemma:additivity}
For all shapes $S$ and all $l,k,l',k'\geq 0$ it holds that $F_{l',k'}(F_{l,k}(S))=F_{l'+l,k'+k}(S)$.
\end{lemma}

\begin{proof}
By definition, $F_{l,k}(S)$ gives a new point $u^\prime=( u_x + (2^l-1)w(C_j,u), u_y + (2^k-1)s(R_i,u) )$ for each $u\in S$ and a rectangular set of points $R(u^{\prime})= \{u^{\prime}_x, u^{\prime}_{x}+1,\ldots, u^{\prime}_{x}+(2^l-1)\} \times \{u^{\prime}_y, u^{\prime}_{y}+1,\ldots, u^{\prime}_{y}+(2^k-1)\}$ originating at $u^\prime$. This gives a new shape $S'$. Applying $F_{l',k'}$ to each $u'$ gives a new point $u''$, such that
\begin{align*}
u''&=( u'_x + (2^{l'}-1)w(C_j,u'), u'_y + (2^{k'}-1)s(R_i,u') )\\
&=(u_x + (2^l-1)w(C_j,u) + (2^{l'}-1)w(C_j,u'), u_y + (2^k-1) s(R_i,u) +\\ &\;\;\;\;\;(2^{k'}-1)s(R_i,u'))\\
&=(u_x + (2^l-1)w(C_j,u) + (2^{l'}-1)2^{l}w(C_j,u), u_y + (2^k-1)s(R_i,u) +\\ &\;\;\;\;\;(2^{k'}-1)2^{k}s(R_i,u))\\
&=(u_x + (2^{l'+l}-1)w(C_j,u), u_y + (2^{k'+k}-1)s(R_i,u)).
\end{align*}
The set of all these points $u''$ is, thus, equivalent to the set of points returned by step (1) of $F_{l'+l,k'+k}(S)$.

For the remaining points, take again any $u\in S$. The rectangle of $u''$ returned by $F_{l'+l,k'+k}(S)$ consists of all points enclosed within the following four corners: $(u^{\prime\prime}_x,u^{\prime\prime}_y)$,
\begin{align}
&(u^{\prime\prime}_x+(2^{l'+l}-1),u^{\prime\prime}_y+(2^{k'+k}-1))=\nonumber\\
&(u_x+(2^{l'+l}-1)w(C_j,u)+(2^{l'+l}-1),u_y+(2^{k'+k}-1)s(R_i,u)+(2^{k'+k}-1))=\nonumber\\
&(u_x+(2^{l'+l}-1)(w(C_j,u)+1),u_y+(2^{k'+k}-1)(s(R_i,u)+1)),\label{eq:diagonal-point}
\end{align}
$(u^{\prime\prime}_x, u_y+(2^{k'+k}-1)(s(R_i,u)+1))$, and $(u_x+(2^{l'+l}-1)(w(C_j,u)+1),u^{\prime\prime}_y)$.

We have already shown that the point returned by $F_{l',k'}(F_{l,k}(S))$ for $u$ is the same to the one returned by $F_{l'+l,k'+k}(S)$, that is, $(u^{\prime\prime}_x,u^{\prime\prime}_y)$. It is sufficient to show that the opposite diagonal point returned by $F_{l',k'}(F_{l,k}(S))$ for $u$ is the same as the point of Eq. \ref{eq:diagonal-point}.

Applying $F_{l,k}(S)$ gives for $u\in S$ a new point $u'=(u_x + (2^l-1)w(C_j,u), u_y + (2^k-1)s(R_i,u))$ and the top-right corner of the rectangle of $u'$ is the point $v'=(u^{\prime}_x+2^l-1,u^{\prime}_y+2^k-1)=(u_x+(2^l-1)(w(C_j,u)+1),u_y+(2^k-1)(s(R_i,u)+1))$. Next, applying $F_{l',k'}$ to $S'$ gives for $v'$ a new point $v''$, whose $x$ coordinate is
\begin{align}
v^{\prime\prime}_x &= v^{\prime}_x+(2^{l'}-1)w(C_j,v')\nonumber\\
&= v^{\prime}_x+(2^{l'}-1)(2^l w(C_j,u)+(2^l-1))\label{eq:v''-x-coordinate}
\end{align}
and its $y$ coordinate is $v^{\prime\prime}_y = v^{\prime}_y+(2^{k'}-1)(2^k s(R_i,u)+(2^k-1))$. The top-right corner $w$ of the rectangle of $v^{\prime\prime}$ is obtained by adding $2^{l'}-1$ to both these coordinates, thus, yielding
\begin{align}
w_x &= v^{\prime}_x+2^l(2^{l'}-1)(w(C_j,u)+1)\nonumber\\
&= u^{\prime}_x+(2^l-1)+2^l(2^{l'}-1)(w(C_j,u)+1)\nonumber\\
&= u_x+(2^l-1)(w(C_j,u)+1)+2^l(2^{l'}-1)(w(C_j,u)+1)\nonumber\\
&= u_x+(2^{l'+l}-1)(w(C_j,u)+1)\label{eq:w-x-coordinate}
\end{align}
and, similarly, $w_y = u_y+(2^{k'+k}-1)(s(R_i,u)+1)$. Thus, $w$ is the same point as the one returned by Eq. \ref{eq:diagonal-point}. 

The lemma now follows by observing that the rectangle formed by $F_{l',k'}(F_{l,k}($ $S))$ within the area defined by the points $(u^{\prime\prime}_x,u^{\prime\prime}_y)$, $(u^{\prime\prime}_x,w_y)$, $(w_x,w_y)$, and $(w_x,u^{\prime\prime}_y)$ is missing no points.
\qed 
\end{proof}

\begin{theorem}\label{prop:generlize-local-rectangle}
Given any initial shape $S_I$ and any sequence of $l$ east and $k$ north full doubling operations, the obtained shape is $S_F = F_{l,k}(S_I)$.
\end{theorem} 

\begin{proof}
We will prove by induction that applying any sequence of $l$ east and $k$ north full doubling operations on shape $S_I$ results in a shape $S_F$ which is equivalent to the output of the reconfiguration function $F_{l,k}(S_I)$.

For the base case, let first $l=1$ and $k=0$, that is, a horizontal full doubling operation applied to $S_I$. By the definition of this operation (see Definition~\ref{def:full-doubling}), the obtained shape is $S^\prime_I=\bigcup_{1\leq j\leq C} (\stackrel{j-1}\rightarrow S_{I\;j,\cdot} ~ \cup \stackrel{j}\rightarrow S_{I\;j,\cdot})$. So, given any $u\in S_I$, say in the $j$th column of $S_I$, $u$ is translated right $j-1$ times, thus yielding a new point $u^\prime=(u_x + j-1,u_y)=(u_x+w(C_j,u),u_y)$, as required by (1) of Definition \ref{def:reconf-func}. Moreover, $u$ is translated right $j$ times to give another new point $u''=(u_x + j,u_y)=(u'_x + 1,u_y)$, as required by (2) of Definition \ref{def:reconf-func} for sets $\{u'_x,u'_x + 1\}\times\{u_y\}$. Thus, $S^\prime_I=F_{1,0}(S_I)$ holds. A similar argument for $F_{0,1}(S_I)$ completes the proof of the base case.

For the induction hypothesis, let us assume that after $l$ east and $k$ north full doubling operations, starting from $S_I$, the obtained shape is $S=F_{l,k}(S_I)$. For the inductive step, we first consider the $(l+1)$th east operation is applied to $S$ to give a new shape $S^\prime$. By replacing in the base case $S_I$ by $S$, we get that $S'=F_{1,0}(S)=F_{1,0}(F_{l,k}(S_I))=F_{l+1,k}(S_I)$, where the second equality follows from the inductive hypothesis and the last equality from additivity of the reconfiguration function (Lemma \ref{lemma:additivity}).

For the second argument of $F_{l,k+1}(S_I)$, let us assume the $(k+1)$th north operation is applied to $S$ to give a new shape $S^\prime$. By replacing in the base case $S_I$ by $S$, we get that $S'=F_{0,1}(S)=F_{0,1}(F_{l,k}(S_I))=F_{l,k+1}(S_I)$, where the second equality follows from the inductive hypothesis and the last equality from additivity of the reconfiguration function (Lemma \ref{lemma:additivity}), and this ends the proof.
\qed
\end{proof}

\section{RC Doubling }\label{sec:RC-doubling}
After a formal definition of the \emph{RC doubling} operation, in this section, we study both the class characterization and the {\sc ShapeConstruction} problems. In particular, we develop a linear-time centralized algorithm to decide the feasibility of constructing $S_F$ from $S_I$ and to return a constructor of $S_F$ from $S_I$ if one exists, both within $O(\log n)$-time-steps.

\begin{definition}[RC Doubling]\label{def:RC-operation}
A row and column doubling is a growth operation where a direction $d \in \{$east, west$\}$ ($d \in \{$north, south$\}$) is fixed and all nodes of a subset of the columns (rows, respectively) of shape $S$ generates a new node in $d$ direction.

We define an \emph{RC doubling operation} for columns in the east direction and the other cases can be similarly defined. The operation is applied to a shape $S$ and will yield a new shape $S'$. Let $J$ be the set of indices (ordered from west to east) of all columns of $S$ and $D$ its subset of indices of the columns to be doubled by the operation. For any $j\in J$, let $w(D,j)=|\{j'\in D\;|\; j'<j\}|$, i.e. $w(D,j)$ denotes the number of doubled columns to the west of column $j$. Then the new shape $S'$ is defined as:
\begin{equation*}
S'=(\bigcup_{j\in J} \stackrel{w(C_j)}\rightarrow C_j)\cup(\bigcup_{j\in D} \stackrel{w(C_j)+1}\rightarrow C_j)   
\end{equation*}
That is, every doubled column $C_j$, for $j\in D$, generates a copy of itself to the east. The result is that every column $C_j$, for $j\in J$, is translated east by $w(C_j)$ and additionally the final position of the copy of $C_j$, for $j\in D$, is an east $(w(C_j)+1)$ translation of $C_j$.
\end{definition}

\begin{definition}[Single RC Doubling Operation]\label{def:single-column-doubling-operation}
Let $d\in J$ be the index of the single doubled column. Define $S_{\leq C_d}$ ($S_{\geq C_d}$) to be the set of columns to the west (east, resp.) of column $C_d$, inclusive. That is, $S_{\leq C_d}=\bigcup_{j\in J,j\leq d} C_j$ ($S_{\geq C_d}=\bigcup_{j\in J,j\geq d} C_j$, resp.). Then,
\begin{equation*}
S'=S_{\leq C_d}\cup(\stackrel{1}\rightarrow S_{\geq C_d}).  
\end{equation*}
\end{definition}

Fig.~\ref{fig:single-c-double} and~\ref{fig:single-c-double2} illustrate Definition~\ref{def:single-column-doubling-operation}. 

\begin{figure}[ht]
\centering
\includegraphics[width=3.5in]{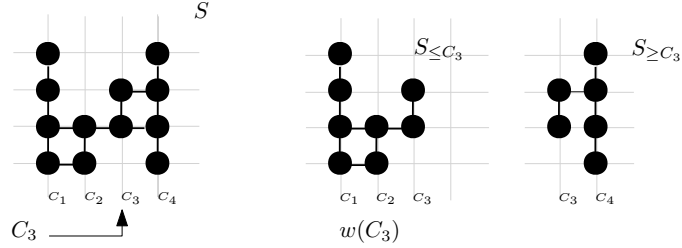}
\caption{An example of a single RC doubling operation (Definition~\ref{def:single-column-doubling-operation}), where $C_3$ is the selected column to be doubled in the \emph{east} direction.}
\label{fig:single-c-double}
\end{figure}

\begin{figure}[ht]
\centering 
\includegraphics[width=3.5in]{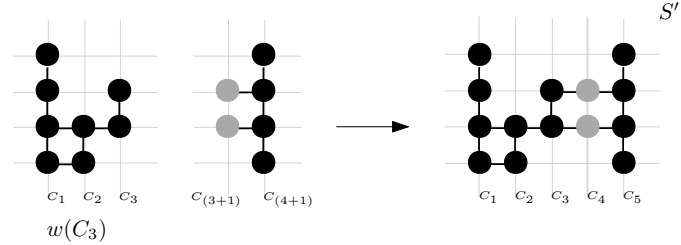}
\caption{The final shape $S'$ after a \emph{single RC doubling operation}. }
\label{fig:single-c-double2}
\end{figure}

\begin{proposition}[Serializability of Parallel Doubling]\label{pro:serializability}
A shape $S_F$ can be generated from a shape $S_I$ through a sequence of RC (parallel) doubling operations iff it can be generated through a sequence of single row/column doubling operations.
\end{proposition}

\begin{proof}
The ``if'' part follows trivially, because single row/column doubling is a special case of $RC$ (parallel) row/column doubling.

We now prove the ``only if'' part. That is, we will show that for any sequence $\sigma$ of $RC$ doubling operations generating $S_F$ from $S_I$, there is a sequence of single $RC$ operations $\sigma^\prime$ generating $S_F$ from $S_I$. For every $RC$ row/column operation $o$ in $\sigma$, we add to $\sigma^\prime$ a break down of $o$ into a sequence of individual operations of the columns (rows, resp.) in the subset. We do this for a west to east (bottom to top, resp.) ordering of the columns (rows, resp.), even though any permutation would do equally well.
We will prove the equivalence for east column operations as all other operations simply follow by rotating the whole system by 90\degree, 180\degree, and 270\degree. Let $D=\{d_1,d_2,\ldots,d_{|D|}\}$ be the set of column indices in an RC column operation applied to a shape $S$ and denote by $o_D$ the operation and by $o_D(S)$ the shape obtained by applying $o_D$ to $S$. Similarly, denote by $o_d(S)$ the shape obtained by applying a single doubling of column $C_d$, for $d\in D$, to shape $S$. Then it is sufficient to show that
\begin{equation*}
o_D(S)=o_{d_{|D|}}(o_{d_{|D|-1}}(o(\ldots o_{d_2}(o_{d_1}(S))\ldots))).  
\end{equation*}
It is sufficient to show that the right hand side of the above equation translates every column $C_j$ of $S$, $w(C_j)$ times (the copies of columns $C_d$, $d\in D$, will then be trivially translated by $w(C_j)+1$ as also done by the left hand side of the equation).

By definition, $o_d(S)=S_{\leq C_d}\cup(\stackrel{1}\rightarrow S_{\geq C_d})$ translates all columns to the east of $C_d$ by 1. Therefore, any given column $C_j$ of $S$, will be translated by 1 for each of the column operations of $o_{d_{|D|}}(o_{d_{|D|-1}}(o(\ldots o_{d_2}(o_{d_1}(S))\ldots)))$ that happen to its west. The number of those is equal to the number of indices in $D$ which are less than $j$, which is, by definition, equal to $w(C_j)$.
\qed 
\end{proof}

\begin{definition}[Consecutive Column/Row Multiplicities]\label{def:cons-columns}
Given a shape $S$ and a column $C_j$ (row $R_i$) of $S$ which is either the leftmost column (bottom-most row) (i.e., $j=1$) or $C_{j-1}\neq C_j$ ((i.e., $i=1$) or $R_{i-1}\neq R_i$) (where equality is defined up to horizontal (vertical) only translations of columns (rows)), the multiplicity $M_S(C_j)$ ($M_S(R_i)$) of column (row) $C_j$ ($R_i$), is defined as the maximal number of consecutive identical copies of $C_j$ ($R_i$) in $S$ to the right (top) of $C_j$ ($R_i$) (inclusive).
\end{definition}

\begin{definition}[Baseline Shape]\label{def:base-shape}
The \emph{baseline shape} $B(S)$ of a shape $S$, is the shape obtained as follows. For every column $C_j$ of $S$ with $M_S(C_j)>1$, remove all consecutive copies of $C_j$ to its right (non-inclusive) and compress the shape to the left to restore connectivity. Then for every row $R_i$ of $S$ with $M_S(R_i)>1$, remove all consecutive copies of $R_i$ to its top (non-inclusive) and compress the shape down to restore connectivity. Observe that all columns and rows of $B(S)$ have multiplicity 1. Moreover, any shape whose columns and rows all have multiplicity 1 is called a baseline shape. 
\end{definition}

\begin{obs}
Every shape $S$ that has a given shape $B$ as its baseline can be obtained by successively doubling the original columns and rows of $B$, thus, creating consecutive multiplicities of these columns and rows. Focusing on any given column $C$ (row $R$) of $B$, we denote by $M_S(C)$ ($M_S(R)$) the consecutive multiplicity of that particular column (row) of $B$ in $S$ (and not its total multiplicity, in case two or more identical copies of a column or row exist in non-consecutive coordinates of $B$). 
\end{obs}

\begin{theorem}\label{the:RC-characterization}
A shape $S_I$ can generate a shape $S_F$ through a sequence of RC doubling operations iff $B(S_I)=B(S_F)=B$ and for every column $C$ and row $R$ of $B$ it holds that $M_{S_F}(C)\geq M_{S_I}(C)$ and $M_{S_F}(R)\geq M_{S_I}(R)$.
\end{theorem}

\begin{proof}
To prove that the condition is sufficient, we can w.l.o.g. restrict attention to single RC doubling operations (as these are special cases of RC doubling operations). Then, for every column $C$ of $B$ for which $M_{S_F}(C)> M_{S_I}(C)$ holds, we double the west-most copy of column $C$ in $S_I$, $M_{S_F}(C)-M_{S_I}(C)$ times to the east. Similarly, for rows. It is not hard to see that any sequence of these operations applied to $S_I$, yields $S_F$.

For the necessity of the condition, we need to show that if $S_I$ can generate $S_F$ through a sequence of RC doubling operations, then $B(S_I)=B(S_F)=B$ and the multiplicities are as described in the statement. We first observe that, by Proposition \ref{pro:serializability}, $S_I$ can also generate $S_F$ through a sequence of single RC doubling operations. So, it is sufficient to show that violation of any of the conditions would not allow for a valid sequence of single RC doubling operations.

Let us first assume that $B(S_I)=B(S_F)=B$ holds, but $M_{S_F}(C)\geq M_{S_I}(C)$ does not, that is, $M_{S_F}(C)<M_{S_I}(C)$ for some column $C$ of $B$. Then, there is no way of obtaining $S_F$ from $S_I$ as this would require deleting $M_{S_I}(C)-M_{S_F}(C)$ copies of $C$. Similarly, if $M_{S_F}(R)\geq M_{S_I}(R)$ is violated.

Finally, assume that $B(S_I)\neq B(S_F)$ and that $S_I\rightsquigarrow S_F$ still holds. By definition of baseline shapes, $B(S_I)\rightsquigarrow S_I$ and $B(S_F)\rightsquigarrow S_F$ hold, thus, we have $B(S_I)\rightsquigarrow S_I\rightsquigarrow S_F$ and $B(S_F)\rightsquigarrow S_F$. That is, there is a sequence of single column/row operations starting from $B(S_I)$ and another starting from $B(S_F)$ that eventually make the two shapes equal (starting originally from two unequal baseline shapes). So, there must be a pair $\sigma$ and $\sigma^\prime$ of such sequences minimizing the maximum length $\max_{\sigma,\sigma^\prime}(|\sigma|,|\sigma^\prime|)$ until the two shapes first become equal. Call $S_t$ and $S^\prime_{t^\prime}$ the dynamically updated shapes by $\sigma_t$ and $\sigma^\prime_{t^\prime}$, respectively. In what follows we omit the time-step subscripts. Let us assume w.l.o.g. that it is the last step $t_{min}$ of $\sigma$ that first satisfies $S=S^\prime$ and that this step is a doubling of a column $C$. Thus, after step $t_{min}$, both $S$ and $S^\prime$ contain an equal number of at least two consecutive copies of $C$. But the only way a shape can first obtain two consecutive copies of a column is by doubling one of its columns, thus, there must be a previous single column doubling operation in $\sigma^\prime$ that doubled column $C$ (note that, at that point, $C$ could have been a subset of the final version of the column). Deleting that operation from $\sigma^\prime$ and the last operation at $t_{min}$ from $\sigma$, yields a new pair of sequences that satisfy $S=S^\prime$ at some $t\leq t_{min}-1$, thus, contradicting minimality of the $(\sigma, \sigma^\prime)$ pair. We must, therefore, conclude that $S_I\rightsquigarrow S_F$ cannot hold in this case, (see Fig. \ref{fig:reachability}).
\qed
\end{proof}

\begin{figure}[ht]
\centering 
\includegraphics[width=3.5in]{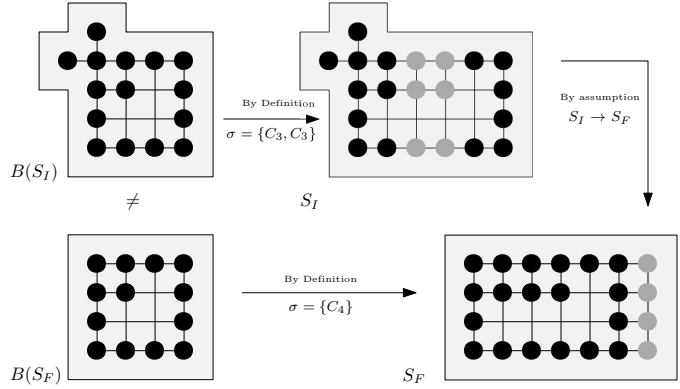}
\caption{Theorem 2 illustration example.}
\label{fig:reachability}
\end{figure}

\begin{proposition}\label{prop:one-Baseline}
For any shape $S$, there is a unique baseline shape $B(S)$. 
\end{proposition}

\begin{proof}
Let us assume a shape $S$ where there are multiple baseline shapes, consider two of those shapes, $B_1(S)$ and $B_2(S)$ where $B_1(S) \neq B_2(S)$, by Definition ~\ref{def:base-shape}, all columns and rows of these baseline shapes have multiplicity 1. Thus, these baseline shapes $B_1(S)$and $B_2(S)$ are equal since they have the same multiplicity of the same column and the same row. Therefore, this is a contradiction, and  there is only one baseline shape $B(S)$.
\qed 
\end{proof}

\begin{lemma}\label{lemma:algo-running-time}
For any $S_I$, $S_F$ satisfying the conditions of Theorem \ref{the:RC-characterization}, there is a constructor from $S_I$ to $S_F$ using at most $2\log n$ time-steps, where $n$ is the total number of nodes in $S_F$.
\end{lemma}

\begin{proof}
Since there is a constructor from $S_I$ to $S_F$, then, by Theorem \ref{the:RC-characterization}, $B(S_I)=B(S_F)=B$ and for every column $C$ and row $R$ of $B$ it holds that $M_{S_F}(C)\geq M_{S_I}(C)$ and $M_{S_F}(R)\geq M_{S_I}(R)$. By Definition \ref{def:reachability} (in Section \ref{sec:Models and Preliminaries}) $S_I\rightsquigarrow S_F$, $S_F$ can be obtained by applying on every column $C$ and row $R$ of $S_I$ as many \emph{RC doubling operations} as required to make its multiplicity equal to $S_F$. W.l.o.g. we only show this process applied to columns.

Let $C$ be a column of $B$. Starting from $M_{S_I}(C)$ copies of $C$ in $S_I$ we want to construct the $M_{S_F}(C)$ copies of $C$ in $S_F$. Note that neither $M_{S_F}(C)$ nor $M_{S_F}(C)-M_{S_I}(C)$ are necessarily powers of $2$. Then, let $2^k$ be the greatest power of $2$, such that $M_{S_I}(C)2^k < M_{S_F}(C)$, i.e., $M_{S_I}(C)2^k < M_{S_F}(C) < M_{S_I}(C)2^{k+1}$. 
Then, from the second inequality, it holds that 
$M_{S_F}(C) - M_{S_I}(C)2^k $ $<M_{S_I}(C)2^{k+1}$ $- M_{S_I}(C)2^k$ 
and this leads to 
$ M_{S_F}(C)- M_{S_I}(C)2^k < M_{S_I}(C)2^k$,
which means that if we construct $M_{S_I}(C)2^k$ columns then columns remaining to be constructed to reach $M_{S_F}(C)$ will be less than the constructed ones.

So, we construct $M_{S_I}(C)2^k$ columns (including the original column) by always doubling, within $ k \leq \log (M_{S_F}(C))$ steps. Once we have those, we double in one additional time-step $M_{S_F}(C) - M_{S_I}(C)2^k$ of those to get a total of $M_{S_F}(C)$ columns within $k + 1 \leq \log (M_{S_F}(C))$ steps. If we set $M_{S_F}(C)$ to be the maximum multiplicity of $S_F$, then for every column $C'\neq C$, its multiplicity $M_{S_F}(C') \leq M_{S_F}(C)$ can be constructed in parallel to the multiplicity of $C$, thus, within these $\log (M_{S_F}(C))$ steps. And similarly for rows. As $M_{S_F}(C) \leq n$ and $M_{S_F}(R)\leq n$, where $n$ is the number of nodes of $S_F$, it holds that all column and row multiplicities can be constructed within at most $2\log n$ time-steps.
\qed
\end{proof}

\begin{obs}
To construct any shape $S$ of size $n$, it requires at least $\lceil \log n \rceil $ time-steps.
\end{obs}

We now present an informal description of a linear-time algorithm for {\sc ShapeConstruction}. The algorithm decides whether a shape $S_F$ can be constructed from a shape $S_I$ and, if the answer is positive, it returns an $O(\log n)$-time-step constructor.\\

After that, Algorithm~\ref{alg:baseline-shape},~\ref{alg:decision} and~\ref{alg:trans} shows the pseudo code that briefly formulates this procedure. Given a pair of shapes $(S_I,S_F)$, do the following: \noindent Given a pair of shapes $(S_I,S_F)$, do the following:
\begin{mylist}{Step}
\item Determine the \emph{baseline shapes} $B(S_I)$ and $B(S_F)$ of $S_I$ and $S_F$, respectively. Then compare $B(S_I)$ with $B(S_F)$ and, if they are equal, proceed to Step 2, otherwise return \emph{No} and terminate. 

\item Since we have $B = B(S_I) = B(S_F)$, if for all columns $C$ (rows $R$) of $B$ it holds that $M_{S_I}(C) \leq M_{S_F}(C)$ and $M_{S_I}(R)\leq M_{S_F}(R)$ then proceed to Step 3, else return \emph{No} and terminate.
\item Output the constructor defined by Lemma \ref{lemma:algo-running-time}.
\end{mylist}

The next Algorithm~\ref{alg:baseline-shape} does not depend on
the order of removing consecutive columns/rows multiplicities or whether compress columns first or rows.
\SetAlgoNoLine
\begin{algorithm}[ht]
\caption{Baseline Shape B(S) }
\label{alg:baseline-shape}
\SetKwInOut{KwIn}{Input}
    \SetKwInOut{KwOut}{Output}
    \SetKwFunction{Function}{Baseline-Shape}
    \SetKwProg{Fn}{Function}{:}{}
    \KwIn{$S$.}
    \KwOut{$B(S)$. }
    \tcc{Baseline function to check the multiplicity of every column and row of a shape $S$ and return the $B(S)$ after remove duplication.}
    \SetKwProg{Pn}{Function}{:}{\KwRet}
    \Pn{\Function{$S$}}{
         \For{every column ($C_j$ of $S$)}
        {
        
        \If {the multiplicity $M_S(C_j)$ of column $C_j$ is more than one}
           { ($C_j \leftarrow C_{j+1}$)}
            
        }
        \For{every row ($R_i$ of $S$)}
        {
        
        \If {the multiplicity $M_S(R_i)$ of row $R_i$ is more than one}
           { ($R_i \leftarrow R_{i+1}$)}

        }
        \KwRet $B(S)$
    }
\end{algorithm}

The following Algorithm~\ref{alg:decision} determines the feasibility of growing $S_I$ to $S_F$, and then we use that decision to compute the constructor $\sigma$ in Algorithm~\ref{alg:trans}.
\begin{algorithm}[ht]
\caption{Decision}
\label{alg:decision}
\SetKwInOut{KwIn}{Input}
    \SetKwInOut{KwOut}{Output}
    \KwIn{($S_I, S_F$).}
    \KwOut{True if $S_F$ can be obtained from $S_I$ via a sequence of $RC$ doubling operations; False otherwise. }
       Compute the baseline of $S_I$, \Function{$S_I$}\\
       Compute the baseline of $S_F$, \Function{$S_F$}\\
        \If{$B(S_I)$ is equal to $B(S_F)$}
             {\KwRet True}
         \Else
            {\KwRet False}
            
\end{algorithm}
\begin{algorithm}[ht]
\caption{Constructor $S_I\stackrel{\sigma}\rightsquigarrow S_F$ }
 \label{alg:trans}
    \SetKwInOut{KwIn}{Input}
    \SetKwInOut{KwOut}{Output}
    \KwIn{Decision, from Algorithm~\ref{alg:decision}.}
    \KwOut{Constructor $\sigma$.}
    
         \While {Decision}
         {\tcc{Find the constructor $\sigma$ for columns.}
         \For { every column ($C_j$ of $S_F$)}
            {\For { every column ($C_j$ of $S_I$)}{
            \If{the index of every column $C_j$ of $S_F$ and $S_I$ are equal }
            {
            \If {multiplicity of $M_{S_F}(C_j)$ are greater than $ M_{S_I}(C_j)$}
                {
                
                compute the difference in $m(C_j)$ variable\\
                
                append ($m(C_j)$ to $\sigma$)
                }
                }
            }
            
            }
           
            }
                compute the maximum of  $\sigma $ in $max_{value}$ variable.\\
                \tcc{Count the steps $k$ for doubling columns.}
                \If{$k= \log$ $max_{value}$ $\bmod 2$ does not equal $0$}
                {
                add one extra step to $k$ variable to double the reminder of columns that are not power of 2.
                }
                \KwRet{$\sigma=\{m(C_{1}),m(C_{2}),\ldots,m(C_{j})\}$.}
\end{algorithm}

Finally, together Proposition~\ref{pro:serializability}, Theorem~\ref{the:RC-characterization},  and Lemma~\ref{lemma:algo-running-time} imply that:

\begin{theorem}\label{the:poly-algo}
Algorithm~\ref{alg:trans} is a linear-time algorithm for {\sc ShapeConstruction} under RC doubling operations. In particular, given any pair of shapes $(S_I,S_F)$, when $(S_I \rightsquigarrow S_F)$ the algorithm returns a constructor ${\sigma}$ of $S_F$ from $S_I$ of $O(\log n)$-time-steps.
\end{theorem}

\begin{proof}
In order to prove the running time of the constructor $S_I\stackrel{\sigma}\rightsquigarrow S_F$ is linear, we will analyse the pseudo code of Algorithm~\ref{alg:baseline-shape},~\ref{alg:decision} and~\ref{alg:trans}. First, in Algorithm~\ref{alg:baseline-shape}, we determine the baseline shape $B(S)$ by comparing every column $C_j$ of shape $S$ from left to right, that is, $C_j$ to $C_{j+1}$, $C_{j+2}$ until we find the first one $C_{j+x}$ which is not equal to $C_j$. Then, we start from $C_{j+x}$ and do the same for $C_{j+x+1}$, $C_{j+x+2}$ until again we find the first one which is not equal to $C_{j+x}$ and so on. Thus, every column is involved to at most one comparison to a column to its right, that is, the number of comparisons is $C_{|J|-1}$, the same argument holds for rows.

In order to determine the multiplicities of columns, we actually compare two columns $C_j$ and $C_{j+x}$ point-to-point until either the points are exhausted or the first pair of unequal points is found, which means that every point $(u_x,u_y) \in C_j$ is involved in at most one comparison to another point. Therefore, the total number of comparisons equals $|S|$, and the total time-steps is linear.

For Algorithm~\ref{alg:decision}, we compare $B(S_I)$ and $B(S_F)$, by translating them into the same origin and assigning the bottom-most point of their leftmost column to $(0,0)$ and accordingly updating the coordinates of all other points. This procedure yields a decision which is based on Algorithm~\ref{alg:baseline-shape}, which is also linear. 
Following that Algorithm~\ref{alg:trans}, there is a variable $m(C_j)$ for each column $C_j$, that keeps track of the number $C_j$ copies, and if the comparison of $C_j$ and $C_{j+x}$ is true, we increase this variable. The total number of increments for $m(C_j)$ is at most equals the number of columns $C_j$, thus, $C_J$. Then, at each point of these duplicated columns $m(C_j)$, we subtract from it at most twice (once for columns and once for rows), for a total of two operations per point. Therefore, the overall running time for returning a sequence $\sigma$ for the growth process is equal at most the number of nodes of $S$.
\qed 
\end{proof}

\section{Doubling }\label{sec:doubling}
This section studies \emph{doubling} operations in their most general form, where up to individual nodes can be involved in a growth operation. We start with a formal definition of two sub-types of general doubling operations and then investigate both the class characterization and {\sc ShapeConstruction} problems. By focusing on the special case of a singleton $S_I$, we give a universal linear-time-step (i.e., slow) constructor and, on the negative side, prove that some shapes cannot be constructed in sub-linear time-steps. Our main results are then two universal constructors that are efficient (i.e., polylogarithmic time-steps) for large classes of shapes. Both constructors can be computed by polynomial-time centralized algorithms for any input $S_F$.

Given a shape $S$ and two neighboring nodes $u, v\in S$, let $S(u)$ and $S(v)$ be the maximal connected sub-shapes of $S$ containing $u$ but not $v$ and $v$ but not $u$, respectively. When $u$ is doubling in the direction of $v$, call that direction $d$, rigidity of connections (see Definition~\ref{def:rigid-connection}) implies that any $w\in S(u)\setminus S(v)$ must remain in its position while any $z\in S(v)\setminus S(u)$ must translate by 1 in direction $d$. 
For any node in $S(u)\cap S(v)$ these two actions would contradict each other. Such nodes belong to a $u,v,\ldots,u$ cycle, and any such cycle must break or grow  (if more than one nodes are allowed to double)  in at least one of its connections, in addition to the connection $uv$ which will by assumption grow. In this paper, we focus on the case where all these cycles break (or grow) at the $(C_j,C_{j+1})$ cut. Depending on how we choose to treat such cycles, we shall define two sub-types of general \emph{doubling operations}: \emph{rigidity-preserving doubling} and \emph{rigidity-breaking doubling}. Intuitively, in the former for all affected edges $e$ in the $(C_j,C_{j+1})$ cut a node is generated over $e$, while in the latter any subset of those edges can simply break.

We start with a special case of the rigidity-breaking doubling operation in which, in every time-step, a single node doubles. This special case is particularly convenient for the class characterization problem, as it can provide a (slower but simple) way to simulate both types of doubling operations. It also serves as an easier starting point towards the definition of the more general operations.

\begin{definition}[Single-Node Doubling]\label{def:partial-doubling}
A \emph{single-node doubling operation} is a growth operation in which at any given time-step $t$, a direction $d  \in \{north, east,$ $south, west$\} is fixed and a single node $u$ of shape $S$ doubles in direction $d$.

Consider w.l.o.g. an $east$ doubling operation applied on $u=(u_x,u_y) \in C_j$ of $S$. If $u$ has no east neighbor in $S$, then, $u$ generates a new node $u'=(u_x+1,u_y) \in C_{j+1}$ and the obtained shape is $S\cup \{u'\}$. Otherwise, $u$ has a neighbor $v \in C_{j+1}$ of $S$ which will need to translate by 1 in the \emph{east} direction together with some sub-shape of $S(v)$. 

We identify the maximal connected sub-shape $S^{\prime}(u) \subseteq S(u)$ that contains no node from columns $C_m$, for all $m \geq j + 1$, and the maximal connected sub-shape $S^{\prime}(v) \subseteq S(v) \setminus S^{\prime}(u)$. That is, $S^{\prime}(u)$ contains all nodes on $u$’s side that must stay put, while, from the remaining nodes, $S^{\prime}(v)$ contains all nodes that must translate by 1. Any bicolor edge (one whose one endpoint is in $S^{\prime}(v)$ and the other endpoint in $S^{\prime}(u)$; we call these the \emph{bicolor edges associated with $uv$}) must be an edge of the $(C_j,C_{j+1})$ cut. We remove all bicolor edges in order to perform the operation.
\end{definition}

\begin{definition}[Rigidity-Preserving Doubling Operation]\label{def:rigid-partial-doubling}
A rigidity-pre\-serving doubling operation is a generalization of a single-node doubling operation. In every time-step a direction $d$ is fixed and, for any node $u$ that doubles towards a neighbor $v$ in direction $d$ and for all bicolor edges $e$ associated with $uv$, a node is generated over $e$. (see Fig.~\ref{fig:rigid-v}).
\end{definition}

\begin{figure}[ht]
\centering 
\includegraphics[width=4in]{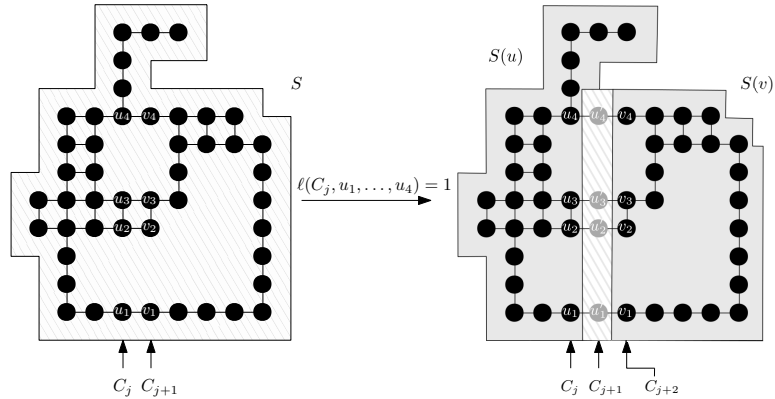}
\caption{An illustration of Definition~\ref{def:rigid-partial-doubling}, in which all nodes of $(C_j)$ must double and the sub-shape $S(v)$ must be shifted to the east by one.}
\label{fig:rigid-v}
\end{figure}

\begin{definition}[Rigidity-Breaking Doubling Operation]\label{def:non-rigid-partial-doubling}
A rigidity-brea\-king doubling operation is a generalization of a single-node doubling operation. In every time-step a direction $d$ is fixed and, for any node $u$ that doubles towards a neighbor $v$ in direction $d$ and for all bicolor edges $e$ associated with $uv$, either a node is generated over $e$ or $e$ is removed (see Fig.~\ref{fig:non-rigid-v} and \ref{fig:break-grow-cycle}).
\end{definition}

\begin{figure}[ht]
\centering 
\includegraphics[width=4in]{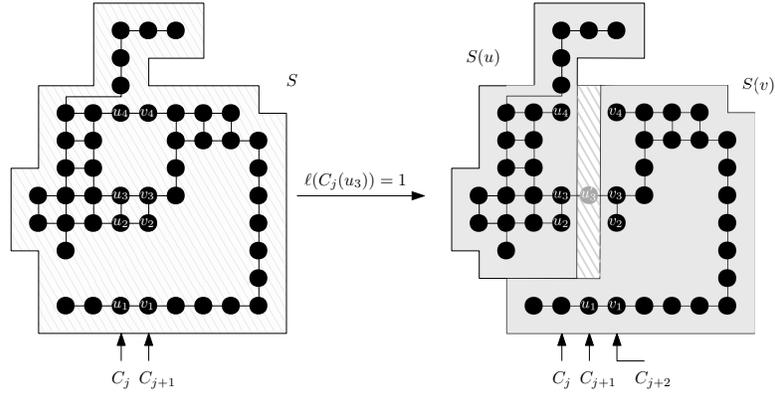}
\caption{An illustration of Definition~\ref{def:non-rigid-partial-doubling}, where there is one node $u_3 \in C_j$ doubles to the east and shifts the connected component in the same direction, while other edges in $C_j$ are removed.}
\label{fig:non-rigid-v}
\end{figure}

\begin{figure}[ht]
\centering 
\includegraphics[width=1.4in]{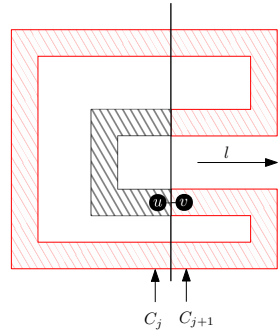}
\caption{An example of a rigidity-breaking doubling operation highlighting also the sets $S'(u)$ and $S'(v)$ and a bicolor edge in the $(C_j,C_{j+1})$ cut (the one above $uv$).}
\label{fig:break-grow-cycle}
\end{figure}

\begin{obs}
Observe that general doubling operations allow for the translation of components which can extend to both sides of the doubling node(s) (for example, the red component in Fig. \ref{fig:break-grow-cycle}). It is easy to see that such a translation of a $S'(v)$ can never be blocked by a component $S'(u)$ which stays put (e.g., the black component in Fig. \ref{fig:break-grow-cycle}). Because if a node in $S'(v)$ had a neighbor in the direction of translation belonging to the component $S'(u)$, then this would violate maximality of $S'(u)$.
\end{obs}

\begin{proposition}\label{prop:linear-construction}
For any shapes $S_I$ and $S_F$, where $S_I\subseteq S_F$, there is a linear-time-step constructor of $S_F$ from $S_I$.
\end{proposition}

\begin{proof}
Consider a pair of shapes $(S_I,S_F)$, to obtain the final shape $S_F$ starting with a single-node initial shape $S_I$, compute a spanning tree as shown in Fig.~\ref{fig:tree-shape}. Then, select any node as its root, this selected root will correspond to the initial single node of $S_I$.
Since a node can generate at most one new node per time-step. Thus, each phase should consist of at most 4 time-steps $t$. Then, in every phase generate the next level of a \emph{breadth first search} of the spanning tree $T$.
Due to the fact that the next \emph{breadth first search} level always concerns positions that are empty in the shape and are adjacent to a node that has already been generated. Therefore, it is clear that it is possible to fill all these positions without the need to push any existing node.

For more general $S_I \subseteq S_F$, we already have the nodes in $S_I$ and the nodes in $S_F \setminus S_I$ must be generated in order to construct $S_F$. We compute a spanning forest $T$ of $S_F \setminus S_I$, each (maximal) connected component of which is a tree $T_i$ spanning that component. For every such tree $T_i$ we pick a neighbor $u_i$ of $T_i$ in $S_I$, which is guaranteed to exist by maximality of the components of $T$. Then from every $u_i$ we, set $u_i$ as the root of $T_i\cup \{u_i\}$ and run the process we already have for singleton $S_I$, in order to construct $T_i$.
\qed 
\end{proof}

\begin{figure}[ht]
\centering 
\includegraphics[width=1.4in]{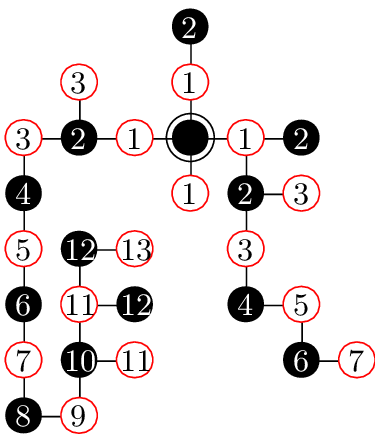}
\caption{Constructing any shape $S_F$ in linear time-steps using BFS.}
\label{fig:tree-shape}
\end{figure}

We call any $L\geq 1$ consecutive nodes connected horizontally or vertically an \emph{$L$-line}.

\begin{proposition}\label{prop:3-line-prop}
If a $3$-line is ever generated, it must be preserved in the final shape $S_F$, that is, rigidity-preserving doubling operation will never break the $3$-line.
\end{proposition}

\begin{proof}
Let us consider a horizontal w.l.o.g. $3$-line, 
since there are two possible directions of rigidity-preserving doubling operation that can be applied to the $3$-line, either  \emph{north} or  \emph{east}, we will first assume a \emph{north} rigidity-preserving doubling operation that breaks the $3$-line into two distinct rows. As illustrated in Fig.~\ref{fig:3-line} (a), there must be a row below the $3$-line to perform this operation and break it.
Since all nodes of the $3$-line belong to the same component and such an assumption of applying this operation and break the $3$-line cannot hold because it contradicts the definition of rigidity-preserving doubling operation (see Definition~\ref{def:rigid-partial-doubling} where all nodes below the 3-line component must push the complete $3$-line to the north and hence never break).

Now, let us assume the other case of applying an \emph{east} rigidity-preserving doubling operation to break the $3$-line, where a break means that there will be an unoccupied position between two consecutive nodes of $3$-line, which means a column $C_j$ splits the line into two parts.(i.e., it must be going through a node of the $3$-line).
Applying this operation by one of the nodes of $3$-line, this will only expand the $3$-line to the \emph{east} and never break it. So, there must be a row below or above the $3$-line. Let us consider it is below w.l.o.g, this implies that the node of $C_j$ who generate to the east and pushed right and the node of the $3$-line breaking to the right must belong to the same component defined in the area to the right of the column, as in Fig~\ref{fig:3-line} (b). Therefore, this contradict with the fact that rigidity-preserving doubling operation when pushing a component must generate nodes in all attachment points, and this ends the proof.
\qed 
\end{proof}

\begin{figure}[ht]
\centering 
\parbox{2.5cm}{
\includegraphics[width=2.5cm]{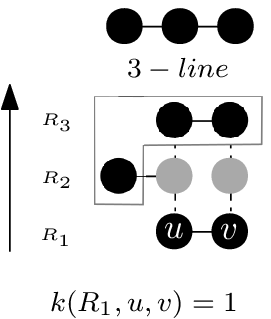}
\caption*{(a)}}
\qquad
\qquad
\qquad
\parbox{2.5cm}{
\includegraphics[width=2.5cm]{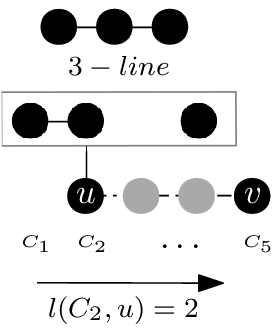}
\caption*{(b)}}
\caption{Maintaining the $3$-line shape in the final construction. }
\label{fig:3-line}
\end{figure}

\begin{definition}[Staircase]\label{def:staircase}
A staircase is a shape $S$, in which each step consists of at least 3 consecutive nodes.
\end{definition}

A staircase is a shape $S$, in which each step consists of at least 3 consecutive nodes, whereas an exact-staircase consists of two nodes.

\begin{proposition}
A staircase of size $n$ requires $\Omega(n)$ time-steps to be generated by rigidity-preserving doubling operations.
\end{proposition}

\begin{proof}
Starting from a shape $|S| = 1$ and performing a rigidity-preserving doubling operation at the first time-step $t_1$, continuing the same operation at the next time-step, we observe that at $t_3$ we have two possible situations, either building a square of $4$ nodes or generating nodes at the two endpoints. The initial attempt of constructing a square will fail because adding any node will result in a $3$-line (see Proposition~\ref{prop:3-line-prop}).

Consequently, we proceed with the second where we reach a point of generating nodes at the two endpoints of the constructed staircase, as shown in Fig.~\ref{fig:linear-staircase}.
As a result, growth can only be at most $2$ nodes it each time-step $t$, which is linear in total.
\qed 
\end{proof}

\begin{figure}[ht]
\centering 
\includegraphics[width=3in]{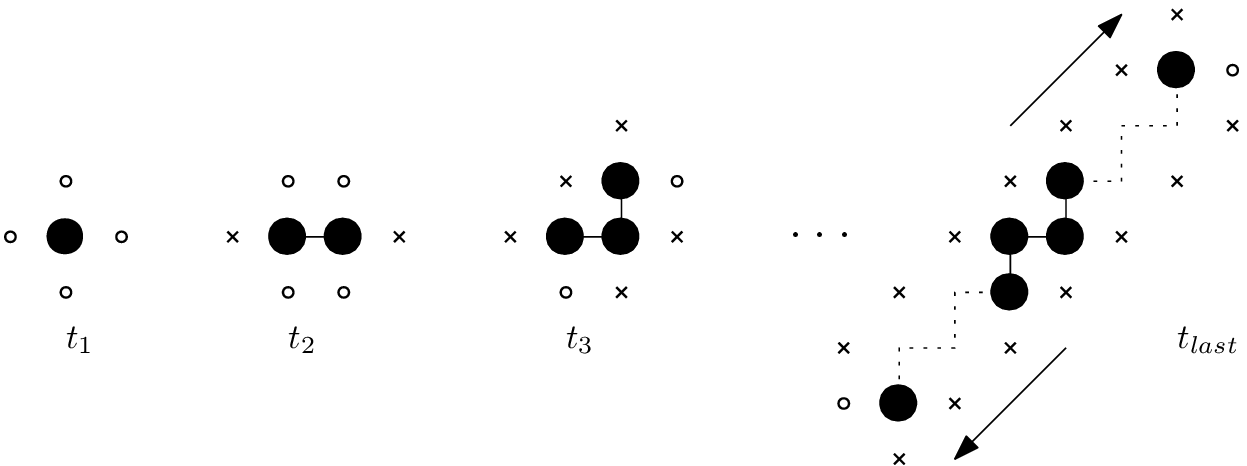}
\caption{Building a staircase in linear time-steps.}
\label{fig:linear-staircase}
\end{figure}

\begin{definition}[Exact-staircase]\label{def:exact-staircase}
An exact-staircase is a special case of staircase shape, in which each step consists of two nodes.
\end{definition}

\begin{proposition}\label{prop:excat-staircase-sublinear}
A rigidity-preserving or rigidity-breaking doubling operation cannot build an exact staircase shape $S$ within a sub-linear time.
\end{proposition}

\begin{proof}
Assume that the last time-step $t_{last}$ of a rigidity-preserving doubling operation will lead to an exact-staircase $S$ of $n$ nodes. We claim that this only occurred if only the two endpoints of the staircase is involved in this operation at $t_{last}$.
To prove the above claim, assume now, for the sake of contradiction, that at least one internal node $u$ was involved in $t_{last}$ of a \emph{north} rigidity-preserving doubling operation, and the final shape constructed was an exact staircase.
Then, node $v$ must have been the node generated by $u$ at step $t_{last}$ and $v$ did not exist in $t_{last-1}$. But this implies that the part of the staircase above $u$ could not have been connected to $u$ in $t_{last-1}$, as shown in Fig.~\ref{fig:excat-staircase-sublinear} (a), because if it were connected above $u$, then the generation of $v$ would have pushed it one row to the north as illustrated in Fig.~\ref{fig:excat-staircase-sublinear} (b). As this does not hold, we must assume that the top part of the staircase cannot have been connected to $u$, and this contradict the fact that the assumed type of constructions cannot break connectivity.

As a consequence, this contradicts the assumption that the internal node $u$ may have involved in $t_{last}$. 
Therefore, only the two endpoints of the staircase can generate in $t_{last}$, as a result, in every step the shape can only grow by at most $2$, giving a $t \geq \lfloor(n/2)\rfloor$.
\qed
\end{proof}

\begin{figure}[ht]
\centering 
\parbox{5cm}{
\includegraphics[width=5cm]{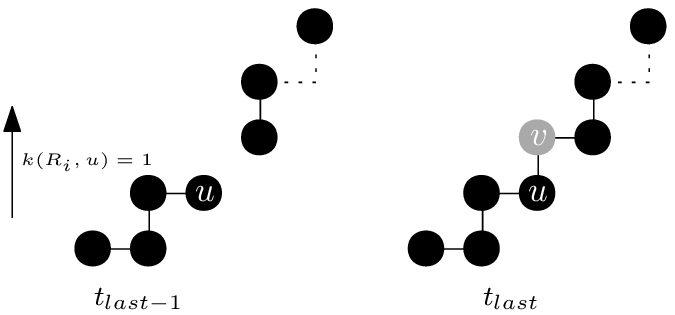}
\caption*{(a)}}
\qquad
\parbox{5cm}{
\includegraphics[width=5cm]{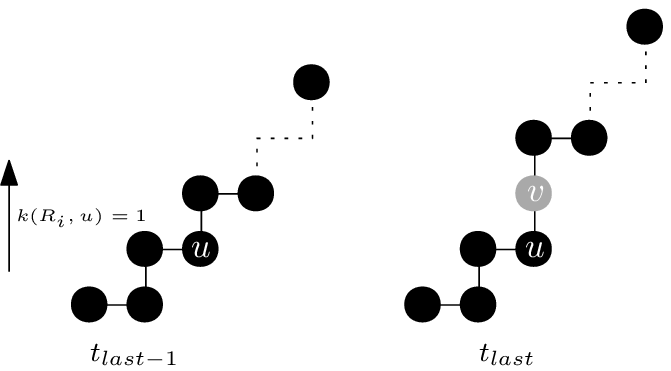}
\caption*{(b)}}
\caption{An explanation of Proposition~\ref{prop:excat-staircase-sublinear}.}
\label{fig:excat-staircase-sublinear}
\end{figure}

\begin{obs}
Any construction that rigidity-preserving  doubling operation performs can also be constructed by rigidity-breaking  doubling but not the opposite.
\end{obs}

Next, by putting together the universal linear-time-steps constructor of Proposition \ref{prop:linear-construction} for doubling and the logarithmic-time-steps constructor of Theorem \ref{the:RC-characterization}  (corresponding to Algorithm \ref{alg:trans}) for RC doubling, we get the following general and faster constructor for doubling.

\begin{theorem} \label{the:doubling-general-constructor1}
Given any connected target shape $S_F$, there is an $[O(|B(S_F)|)+O(\log |S_F|)]$-time-step constructor of $S_F$ from $S_I=\{u_0\}$ through doubling operations. Moreover, there is a polynomial-time algorithm computing such a constructor on every input $S_F$. 
\end{theorem}

\begin{proof}
By Proposition \ref{prop:linear-construction}, the baseline $B(S_F)$ of $S_F$ can be constructed from a singleton $S_I$ within $O(|B(S_F)|)$ time-steps through doubling operations. Then, by Theorem \ref{the:RC-characterization}, $S_F$ can be constructed from its baseline $B(S_F)$ within $O(\log |S_F|-|B(S_F)|))$ time-steps, yielding $O(\log |S_F|)$ time-steps in the worst-case, through RC doubling operations. As RC doubling is a special case of doubling, the latter constructor is also one using doubling operations, thus, the theorem follows. 
\qed
\end{proof}

The constructor of Theorem \ref{the:doubling-general-constructor1} is fast as a function of $n=|S_F|$, when $|S_F|-|B(S_F)|$ is large. For example, for all $S_F$ for which $|B(S_F)|=O(\log |S_F|)$ holds, it gives a logarithmic-time-steps constructor of $S_F$. It is also a fast constructor for all shapes $S_F$ that have a relatively small (geometrically) \emph{similar} shape $S_I$ under \emph{uniform scaling} (which is, in our case, equivalent to getting by full doubling in all directions from $S_I$ to $S_F$). Shape similarity is to be defined for orthogonal (i.e., rectilinear) shapes in a way analogous to its definition for general polygons: two shapes $S_I$ and $S_F$ being \emph{similar} if $S_I$ can be made equal (up to translations) to $S_F$ through uniform scaling. Note that shape similarity can be decided in linear time \cite{manacher76,akl78}. In such cases, $S_I$ can again be constructed in linear time-steps from a singleton, followed by a fast construction of $S_F$ from $S_I$ via full doubling in all directions in a round-robin way.

Finally, we give an alternative constructor, based on a partitioning of an orthogonal shape into the minimum number of rectangles. Note that there are efficient algorithms for the problem, e.g., an $O(n^{3/2}\log n)$-time algorithm \cite{imai1986efficient,keil2000polygon}. These algorithms, given an orthogonal polygon $S$, partition $S$ into the minimum number $h$ of rectangles $S_1,S_2,\ldots,S_h$, ``partition'' meaning a set of pairwise non-overlapping rectangles which are sub-polygons of $S$ and whose union is $S$.

\begin{theorem} \label{the:doubling-general-constructor2}
Given any connected target shape $S_F$, there is an $O(h\log |S_F|)$-time-step constructor of $S_F$ from $S_I=\{u_0\}$ through doubling operations, where $h$ is the minimum number of rectangles in which $S_F$ can be partitioned. Moreover, there is a polynomial-time algorithm computing such a constructor on every input $S_F$. 
\end{theorem}

\begin{proof}
Given $S_F$ we use the $O(n^{3/2}\log n)$-time algorithm of \cite{imai1986efficient} to partition $S_F$ into rectangles. As all shapes considered in this paper are orthogonal polygons, $S_F$ is a valid input to the algorithm, thus, the algorithm returns a partition of $S_F$ consisting of the minimum number $h$ of rectangles, $S_1,S_2,\ldots,S_h$.

We define a graph $G'=(V',E')$ associated with those rectangles. In particular, $V'=\{S_1,S_2,\ldots,S_h\}$ and $E'=\{S_iS_j\;|\;S_i,S_j\in V'$ and a node $u\in S_i$ is an orthogonal neighbor to a node $v\in S_j\}$, that is, the graph $G'$ has a vertex for each rectangle and has an edge between two vertices iff the corresponding rectangles are vertically or horizontally adjacent. Note that $G'$ is a connected graph. If not, then $G'$ would consist of at least two connected components $G'_i$. Each $G'_i$ corresponds to a subset of the rectangles, so that each rectangle belongs to a single $G'_i$. Moreover, the rectangles corresponding to the vertices of $G'_i$ form a partition of a shape $W_i$ which is a subshape of $S_F$. These imply that $S_F=\bigcup_i W_i$, where the $W_i$s are pairwise disconnected shapes, thus contradicting connectivity of $S_F$. So, it must hold that $G'$ is connected.

Therefore, we can compute a spanning tree $T$ of $G'$, which we shall then use to define the constructor of $S_F$. Note that $T$ is a spanning tree of the corresponding rectangles, given their adjacency relation. We set as the root of $T$ a maximum-area rectangle $S'_0$ of the partition. 
The sought constructor is based on a BFS traversal of $T$ starting from $S'_0$. The constructor first constructs $S'_0$ through a constructor of Theorem~\ref{the:poly-algo}, in a number of time-steps logarithmic in $|S'_0|$. Then, for the set $N(S'_0)=\{S'\;|\;S'_0S'\in E'\}$ of rectangles adjacent to $S'_0$, we compute the set of points $c(S',S'_0)=\{(x,y)\;|\;S'\in N(S'_0)$ and $(x,y)$ are the coordinates of a corner-node of $S'$ which is adjacent to a node of $S'_0\}$. For every $S'\in N(S'_0)$ we set one of the points in $c(S',S'_0)$ as the initial point from which $S'$ will be constructed. Then we construct $S'$ directly in its final position by a constructor of Theorem~\ref{the:poly-algo}. Note that the selection of these points and the non-overlapping nature of the partition ensure that these growth processes cannot overlap. 

Whenever a rectangle $S'$ in the next level of the spanning tree can be constructed from more than one initial points (either adjacent to the same or to distinct rectangles of the previous level), then one of those points is chosen arbitrarily. All rectangles of a given level are constructed in parallel, thus the time-steps paid per level of $T$ are the time-steps associated with the largest rectangle of that level. The process continues in this way until the furthest level is constructed.

In the worst case, the $h$ rectangles are constructed sequentially, each, rather crudely, in $O(\log |S_F|)$ time-steps, for a total of $O(h\log |S_F|)$ time-steps.
\qed
\end{proof}
Observe that for those shapes $S_F$ for which $h$ is constant or $O(\log |S_F|)$, the constructor returned by Theorem~\ref{the:doubling-general-constructor2} is of logarithmic or polylogarithmic time-steps, respectively.

\section{Conclusion}\label{sec:conclusion}
In this work, we investigated the graph-growth mechanisms proposed by Mertzios \textit{et al.}~\cite{mertzios2021complexity} within 2D characteristics. Specifically, the growth operations explored in this study are \emph{Full doubling}, \emph{RC doubling} and \emph{doubling}. Two problems has been addressed: class characterization and {\sc ShapeConstruction}. Our main results show that, although studying class characterisation under the full doubling operation is straightforward, it enables us to understand the dynamics of growth processes in their most basic forms. 
In addition, we successfully developed a linear time algorithm that determines the feasibility of growing for any pair of shapes $(S_I,S_F)$ and returns the logarithmic-time-steps constructor. Further, we presented two universal constructors for large classes of shapes that are efficient up to polylogarithmic time-steps.

There is a number of interesting problems that are opened by this work. The obvious first target is to obtain the optimal constructor, and develop an algorithm that optimizes the running time. It would also be interesting to run a simulation for such growth processes or explore these operations within 3D system. Finally, one next step is to develop distributed versions of the growth operations described here.

\section*{Acknowledgements}
 We would like to thank Viktor Zamaraev for the many fruitful discussions during the various stages of developing this work.

\bibliography{bibliography}  
\bibliographystyle{splncs04}
\end{document}